\documentclass{article}
\usepackage{spconf}
\usepackage{graphicx}
\usepackage{epstopdf}
\usepackage{epsfig}
\usepackage{amsmath,bm,algorithm,algorithmic,float}
\usepackage{amssymb}
\usepackage[thmmarks,amsmath,amsthm,hyperref]{ntheorem}
\usepackage{multirow, color}

\newtheorem{Thm}{Theorem}

\newtheorem{Lem}{Lemma}
\newtheorem{Corr}{Corrollary}

\title{Distributed Compressive Sensing: Performance Analysis with\\Diverse Signal Ensembles}
\name{Sung-Hsien Hsieh$^{*,**}$, Wei-Jie Liang$^{*,***}$, Chun-Shien Lu$^{*}$, and Soo-Chang Pei$^{**}$}
\small
\address{$^{*}$Institute of Information Science \& CITI, Academia Sinica, Taipei, Taiwan\\
$^{**}$Graduate Inst. Comm. Eng., National Taiwan University, Taipei, Taiwan\\
$^{***}$Department of Mathematics, National Cheng-Kung University, Tainan, Taiwan}
\normalsize
\begin{document}

\maketitle

\begin{abstract}

Distributed compressive sensing is a framework considering jointly sparsity within signal ensembles along with multiple measurement vectors (MMVs).
The current theoretical bound of performance for MMVs, however, is derived to be the same with that for single MV (SMV)  because no assumption about signal ensembles is made.

In this work, we propose a new concept of inducing the factor called ``Euclidean distances between signals'' for the performance analysis of MMVs.
The novelty is that the size of signal ensembles will be taken into consideration in our analysis to theoretically show that MMVs indeed exhibit better performance than SMV.
Although our concept can be broadly applied to CS algorithms with MMVs, the case study conducted on a well-known greedy solver called simultaneous orthogonal matching pursuit (SOMP) will be explored in this paper.
We show that the performance of SOMP, when incorporated with our concept by modifying the steps of support detection and signal estimations, will be improved remarkably, especially when the Euclidean distances between signals are short.
The performance of modified SOMP is verified to meet our theoretical prediction.
\end{abstract}

\noindent

\vspace{-5pt}
\section{Introduction}\label{sec:intro}
\subsection{Background \& Related Works}\label{sec:background}
Compressive sensing (CS) \cite{Donoho2006, Baraniuk2007} of sparse signals in achieving data acquisition and compression simultaneously has been extensively studied in the past few years.
Conventionally, given a measurement vector, CS shows that a sparse signal can be reconstructed via different solvers such as $\ell_{1}$-minimization \cite{Candes2006-F, Chen2008} or greedy approaches \cite{Tropp2007,Needella2009}.
To further reduce the number of measurements, distributed compressive sensing (DCS) \cite{Baron2005,Duarte2013} is a framework considering jointly sparsity within signal ensembles along with multiple measurement vectors (MMVs).

The model of MMVs is described as follows.
Let $X=[x^{1},x^{2},...,x^{L}]\in \mathbb{R}^{N \times L}$ be the signal ensembles, where $L \geq 1$  is the size of signal ensembles, and let $\Phi^{i} \in \mathbb{R}^{M \times N}$ for $1\leq i\leq L$ be a sensing matrix.
$X$ is called jointly $K$-sparse if $ \left| \bigcup_{i=1}^{L} supp\left(x^{i}\right) \right| = K$, where $supp\left(x^{i}\right)$ returns a support set of $x^{i}$ and $\left|\cdot \right|$ is the cardinality function.
Then, signal sampling is conducted via:
\begin{equation}
y^{i}=\Phi^{i} x^{i}.
\label{Eq: DCS}
\end{equation}
Another common formulation assumes $\Phi=\Phi^{1}=\ldots\Phi^{L}$.
Therefore, let $Y = [y^{1},y^{2},...,y^{L}]$, we have:
\begin{equation}
Y=\Phi X.
\label{Eq: C_MMV}
\end{equation}
Assume that $\Phi^{i}$'s for all $i$'s are drawn from i.i.d distribution.
The main difference between the above two formulations is that $rank(Y) = \min(L,M)$ in Eq. (\ref{Eq: DCS}) but $rank(Y) = \min(rank(X),M)$ in Eq. (\ref{Eq: C_MMV}).

DCS \cite{Duarte2013} shows the fundamental bounds on the number of noiseless measurements such that signals can be jointly recovered based on Eq. (\ref{Eq: DCS}).
In addition, DCS shows that supports can be detected correctly when $L\rightarrow \infty$.
In other words, DCS cannot accurately characterize the relationship between $L$ and the performances of solvers such as SOMP.
\cite{chen06, Jin2013, Hu15} focus on the performance analysis based on Eq. (\ref{Eq: C_MMV}) and show that the performance is proportional to $rank(Y)$ with noiseless measurements.
Nevertheless, when multiple sensors sense the same source with $x^{1}=x^{2}=\ldots=x^{L}$, the performance based on Eq. (\ref{Eq: C_MMV}) will be degraded into SMV due to $rank(Y)=1$.
Thus, we find that if the performance analysis in MMVs does not consider $rank(Y)$ as a factor, the analysis will be same as that in SMV.
For example, \cite{Determe2016} shows the performance of SOMP that is irrelevant to $L$, which is almost the same with OMP (a special case of SOMP with $L=1$) \cite{Mo2012}.
On the contrary, when $x^{1}=x^{2}=\ldots=x^{L}$, the performance based on Eq. (\ref{Eq: DCS}) still is improved \cite{Amelunxen14} since Eq. (\ref{Eq: DCS}) can be reformulated as the following SMV formulation (number of measurements is $ML$ instead of $M$):
\begin{equation}
\hat{y}=A \hat{x},
\label{Eq: stack_MMV}
\end{equation}
where $\hat{y} = \left[ y^{1} ; y^{2} ; \ldots ; y^{L}   \right] \in \mathbb{R}^{ML}$, $A=\left[ \Phi^{1}; \Phi^{2} ; \ldots ; \Phi^{L} \right] \in \mathbb{R}^{ML \times N}$, and $\hat{x}=x^{1}=\ldots=x^{L}$.

\subsection{Motivation}\label{sec:motivation}
The discussions so far motivate us to consider a question: how to characterize the performance of MMVs based on Eq. (\ref{Eq: DCS}), especially when $rank(Y)=1$ or when the relaxed assumption, ``Euclidean distances between signals are nonzero,'' is considered.
It is noted that the Euclidean distances considered in this paper include two parts:
one is $\|x^{i} - x^{*} \|_{2}$ for all $i$'s with $x^{*}= \frac{1}{L}\sum_{i=1}^{L} x^{i}$ and another one is $\|x^{i} \|_{2}$ for all $i$'s.
In particular, the relaxed assumption is practical and occurs in cooperative spectrum sensing \cite{Paysarvi-Hoseini2011}, where MMVs are obtained from different sensors to observe a single signal source (spectrum).
Under the circumstance, when the sensors are too close to each other, the observed signal spectra also are similar, implying that $\epsilon$ is small.

\subsection{Contributions}\label{sec:contributions}
In this paper, we are interested in the performance analysis of MMVs model in Eq. (\ref{Eq: DCS}) with a new factor  ``Euclidean distances between signals''.
Compared with previous works, the imposed factor will lead to the performance that is related to the size $L$ of signal ensembles.
We take SOMP as a case study here even our concept can be generally applied to other greedy algorithms.
More specifically, we present a new mechanism for support detection and derive the sufficient condition of correct support detection.
We show that when the Euclidean distance between signals are short or the signals have the same sign, the new mechanism outperforms conventional method remarkably.
In terms of signal estimation, individual sparse signal is conventionally estimated by its corresponding measurement vector.
In our work, however, we explore a strategy of estimating an individual signal from all measurement vectors and show that this strategy is potential to make support detection possible even when $M < K \leq ML$.


\section{Preliminaries}\label{sec:preliminaries}
For a matrix $H$, we denote its transpose by $H^T$ and its pseudo inverse matrix by $H^{\dag}$. For a set $V$ collecting indices, $H_{V}$ is a  submatrix formed by columns of $H$ with indices belonging to $V$.
 $\mathcal{P}(V)$ is the power set of $V$.
  For a vector $u$, the $i^{\scriptsize\mbox{th}}$ entry of $u$ is $u[i]$. $u_{V} \in \mathbb{R}^{|V|}$ is a vector formed by entries of $u$  with indices belonging to $V$.
$\left\|\cdot\right\|_p$ denotes the $\ell_p$-norm. $sign(u)$ extracts the sign of $u$. $abs(u)$ returns the absolute value of $u$.  In addition, denote $\Omega = \bigcup_{i=1}^{L} supp\left(x^{i}\right) $ as the ground truth of support set.
$\mathcal{N}(0,\sigma^{2})$ denotes a normal distribution with zero mean and variance $\sigma^{2}$.

\section{Main Results}\label{sec:our results}
To induce the new factor ``Euclidean distance between signals'' into the theoretical performance analysis of MMVs and see how many advantages we can have, we take SOMP as a case study here (but keep in mind that our idea can be generally applied to other greedy algorithms).
In the following procedure of SOMP, the steps of support detection and signal estimation contain the original ones ((a) and (c)) and the newly added one ((b) and (d)).
\vspace{-5pt}
\begin{enumerate}
\item Initialization: $t=1$, $S=\{ \  \}$, and $r^{i,t}=y^{i}$ for $i=1,\ldots,L$.
\item Support detection:
\begin{itemize}
\item[(a)] $ I =\arg\!\max_{i} u[i] \text{ with }   u =  \sum_{i=1}^{L} \left| (\Phi^{i})^{T}r^{i,t} \right|$
\item[(b)] $ I =\arg\!\max_{i} u[i] \text{ with }   u =  \left|\sum_{i=1}^{L}  (\Phi^{i})^{T}r^{i,t} \right|.$
\end{itemize}
\item Support update:  $S=S\bigcup \{I \}$.
\item Signal estimation:
\begin{itemize}
\item[(c)] $ \hat{x}^{i} = (\Phi^{i}_{S})^{\dag}y^{i} $ with $i=1,\ldots,L$
\item[(d)] $ \hat{x}^{i} = (A_{S})^{\dag}\hat{y} $ with $i=1,\ldots,L$.
\end{itemize}
\item Residual update:  $r^{i,t+1} = y^{i} - \Phi^{i}_{S}\hat{x}^{i}$ with $i=1,\ldots,L$.
\item If $t=K$, stop and output $\bar{x}^{i}_{S} = (\Phi^{i}_{S})^{\dag}y^{i} $ with $i=1,...,L$; otherwise, $t=t+1$ and go to Step 2.
\end{enumerate}

In the above procedure, SOMP-(a+c) denotes the traditional SOMP by choosing (a) as support detection and (c) as signal estimation.
In contrast, steps (b) and (d) are proposed to accommodate for the conditions that the Euclidean distances between signals are short or the signals have the same sign, as mentioned in Sec. \ref{sec:motivation}.
In the following, we will discuss SOMP-(a+c), SOMP-(b+c), and SOMP-(b+d), respectively.

We first explain why we present (b) as an alternative of (a) in certain situations.
In the first iteration of steps (a) and (b), we expect that $u[j]$ for $j \in \Omega$ is large enough to make support detection correct.
We derive the lower bounds of $u[j]$ for $j \in \Omega$ in steps (a) and (b), respectively, as follows:
$$
\begin{aligned}
 (a): u[j] =&\sum_{i=1}^{L} \left| (\Phi^{i}_{j})^{T}r^{i,1}\right| \\
 & = \sum_{i=1}^{L} \left| x^{i}[j] + \left((\Phi^{i}_{j})^{T}\Phi^{i}_{\Omega}-1\right)x^{i}_{\Omega} \right|\\
 & \geq \sum_{i=1}^{L} \left|x^{i}[j] \right| -   \sum_{i=1}^{L} \left|\left((\Phi^{i}_{j})^{T}\Phi^{i}_{\Omega}-1\right)x^{i}_{\Omega} \right| \\
 (b): u[j] = &\left|\sum_{i=1}^{L}  (\Phi^{i}_{j})^{T}r^{i}_{1}\right| \\
 & = \left|\sum_{i=1}^{L}x^{i}[j] +   \sum_{i=1}^{L}\left((\Phi^{i}_{j})^{T}\Phi^{i}_{\Omega}-1\right)x^{i}_{\Omega} \right|\\
 & \geq  \left|\sum_{i=1}^{L}x^{i}[j] \right| - \left|   \sum_{i=1}^{L}\left((\Phi^{i}_{j})^{T}\Phi^{i}_{\Omega}-i\right)x^{i}_{\Omega} \right| .
\end{aligned}
$$
We observe that if $sign(x^{1})=\ldots=sign(x^{L})$, we have $ \sum_{i=1}^{L} \left|x^{i}[j] \right| =  \left|\sum_{i=1}^{L} x^{i}[j] \right|$, and $ \sum_{i=1}^{L} \left|\left((\Phi^{i}_{j})^{T}\Phi^{i}_{\Omega}-1\right)x^{i} \right| \geq \left|   \sum_{i=1}^{L}\left((\Phi^{i}_{j})^{T}\Phi^{i}_{\Omega}-I\right)x^{i} \right| $.
It is easy to induce that (b) achieves more accurate support detection than (a) under the case that all signals have the same sign.
We will further integrate this assumption into our performances analysis later.

Second, we discuss why we present (d) as an alternative of (c).
In this paper, steps (b+d) is equivalent to solving Eq. (\ref{Eq: stack_MMV}) when $x^{1}=x^{2}\ldots=x^{L}$.
Compared with (c), (d) is potential to make support detection possible when $M < K \leq ML$ since the number of measurements in Eq. (\ref{Eq: stack_MMV}) is $ML$.
In other words, no matter what $L$ is, there are infinite solutions to the least square problem with $M < K$ and, thus, (c) fails to estimate the signal correctly.
On the other hand, when $\|x^{i} - x^{j} \| \leq \epsilon$ for all $i \neq j$ with small $\epsilon$, (b+d) is no longer formulated as SMV in Eq. (\ref{Eq: stack_MMV}).
On the contrary, we show that SOMP-(b+d) conducted with Eq. (\ref{Eq: DCS}) still works when $M < K \leq ML$.

To begin with the performance analyses of SOMP-(a+c), SOMP-(b+c), and SOMP-(b+d), we first introduce restricted isometric property (RIP) as follows.
\begin{Lem}(\textbf{Consequences of RIP})\cite{Candes2005}\\
Given a matrix $\Phi$, for $I \subset \Omega$, if $\delta_{|I|}(\Phi)<1$, then, for any $x\in \mathbb{R}^{|I|}$, we have
\begin{equation}
(1-\delta_{|I|}(\Phi))\|x\|_2 \leq \|\Phi_I^{T}\Phi_I x\|_2 \leq (1+\delta_{|I|}(\Phi))\|x\|_2.
\label{eq: RIP}
\end{equation}
\end{Lem}
A measurement matrix $\Phi$ is said to satisfy RIP of order $K$ if there exists a restricted isometric constant (RIC) $\delta(\Phi) \in (0, 1) $ satisfying Eq. (\ref{eq: RIP}) for any $K$-sparse signal $x$.

According to RIC, OMP recovers all $K$-sparse vectors provided $\Phi$ satisfies the sufficient condition that $\delta_{K+1}(\Phi) < \frac{1}{\sqrt{K}+1}$ \cite{Mo2012, Wang2012}.
Similarity, traditional SOMP-(a+c) with signal ensembles sensed via Eq. (\ref{Eq: C_MMV}) needs to satisfy $\delta_{K+1}(\Phi) < \frac{1}{\sqrt{K}+1}$ or $\delta_{K}(\Phi) < \frac{\sqrt{K-1}}{\sqrt{K-1}+\sqrt{K}}$ \cite{Determe2016}.
Nevertheless, as mentioned in Sec. \ref{sec:background}, the sufficient condition never contains $L$ due to no assumption about signal ensembles was made.
In addition, it should be noted that this sufficient condition \cite{Determe2016} cannot be applied to SOMP-(a+c) with Eq. (\ref{Eq: DCS}).
On the other hand, DCS focuses on SOMP-(a+c) with Eq. (\ref{Eq: DCS}) \cite{Baron2005,Duarte2013} but it does not prove  such a sufficient condition.
Thus, in addition to conducting analyses for SOMP-(b+c) and SOMP-(b+d), we also provide analysis for SOMP-(a+c).

To induce $L$ into the sufficient condition of SOMP with signal ensembles being sensed via Eq. (\ref{Eq: DCS}), our main results are summarized as the following three theorems.

\begin{Thm}
\label{thm: SOMP-a-c}
Suppose $x^i \in \mathbb{R}^N$ is a $K-$sparse signal sensed via Eq. (\ref{Eq: DCS}) for $i = 1, \dots, L$ and $\Phi^{i}$'s satisfy RIP.
Then, the SOMP-(a+c) algorithm will perfectly reconstruct $x^{i}$'s if
\begin{equation}
\label{eq: sc of SOMP-a-c}
\sum_{i=1}^{L}\frac{\epsilon_{1}\delta_{K+1}^{2}(\Phi^{i})-(\sqrt{K}+2\epsilon_{1})\delta_{K+1}(\Phi^{i}) +\epsilon_{1}     }{1-\delta_{K+1}(\Phi^{i})}  > 0,
\end{equation}
where $\displaystyle  \epsilon_{1} = \max_{U \in \mathcal{P}(\Omega)\backslash{\emptyset} } \frac{\min_j\|x^{j}_{U}\|_2}{\max_j\|x^{j}_{U}\|_2}$.
\end{Thm}
\begin{proof}
Please see Appendix for detailed proof.
\end{proof}

\begin{Thm}
\label{thm: SOMP-b-c}
Let $\displaystyle A=\frac{1}{\sqrt{L}}[\Phi^{1};\Phi^{2};\ldots;\Phi^{L}]$ and let  $\delta^{max}_{K}=\max_{i} \delta_{K}(\Phi^{i})$.
Suppose $x^i \in \mathbb{R}^N$ is a $K-$sparse signal sensed via Eq. (\ref{Eq: DCS}) for $i = 1, \dots, L$ and $\Phi^{i}$'s satisfy RIP.
Then, the SOMP-(b+c) algorithm will perfectly reconstruct $x^{i}$'s if
\begin{equation}
\label{eq: sc of SOMP-b-c}
(\sqrt{K}+1)\delta_{K+1}(A)+ (1+(\sqrt{K}+1)L\epsilon_2)\delta_{K+1}^{max} < 1,
\end{equation}
where $\displaystyle  \epsilon_{2} = \max_{U \in \mathcal{P}(\Omega)\backslash{\emptyset} } \frac{\sum_{i=1}^{L}\| x^{i}_{U} - x^{*}_{U}\|_{2} }{L\| x^{*}_{U} \|_{2}}$, $x^{*} = \frac{1}{L}\sum_{i=1}^{L}x^{i}$.
\end{Thm}
\begin{proof}
Please see Appendix for detailed proof.
\end{proof}

\begin{Thm}
\label{thm: SOMP-b-d}
Let $\displaystyle A=\frac{1}{\sqrt{L}}[\Phi^{1};\Phi^{2};\ldots;\Phi^{L}]$ and let  $\delta^{max}_{K}=\max_{i} \delta_{K}(\Phi^{i})$. Suppose $x^i \in \mathbb{R}^N$ is a $K-$sparse signal sensed via Eq. (\ref{Eq: DCS}) for $i = 1, \dots, L$, $A$ satisfies RIP, and $K\leq M$.
Then, the SOMP-(b+d) algorithm will perfectly reconstruct $x^{i}$'s with $i=1,\ldots,L$ if
\begin{equation}
\label{eq: sc of SOMP-b-d}
\sqrt{K}(1+L^{2}\epsilon_{3})\delta_{K+1}(A)+ (1+L\epsilon_{3})\delta_{K+1}^{max} < 1,
\end{equation}
where $\displaystyle \epsilon_{3} = \max_{U \in \mathcal{P}(\Omega)\backslash{\emptyset} } \frac{\sum_{i=1}^{L}\| x^{i} - x^{*}\|_{2} }{L\| x^{*}_{U} \|_{2}}$ with $x^{*} = \frac{1}{L}\sum_{i=1}^{L}x^{i}$.
\end{Thm}
\begin{proof}
Please see Appendix for detailed proof.
\end{proof}

In the above three theorems, $\epsilon_{1}$, $\epsilon_{2}$, and $\epsilon_{3}$ describe the characteristics of involved signal ensembles, respectively.
First, among them, Theorem \ref{thm: SOMP-a-c} shows that when the entries of $x^{i}$'s have the same energy (unrelated to $sign(x^{i})$'s), we have $\epsilon_{1}=1$ and SOMP-(a+c) performs best.
On the other hand, the analysis is derived for SOMP-(a+c) without considering signal ensembles as follows.
\begin{Corr}
\label{prop: SOMP-a-c-another ver}
Let $\displaystyle \delta^{max}_{K}=\max_{i} \delta_{K}(\Phi^{i})$. Other assumptions follow Theorem \ref{thm: SOMP-a-c}. Then, the SOMP-(a+c) algorithm will perfectly reconstruct $x^{i}$'s with $i=1,\ldots,L$ if
$$\delta_{K+1}^{max} < \frac{1}{\sqrt{K} +2}.$$
\end{Corr}
\begin{proof}
Detailed proof is skipped due to limited space.
\end{proof}

In comparison with Theorem \ref{thm: SOMP-a-c}, the result of Corollary \ref{prop: SOMP-a-c-another ver} even is worse than SMV since $\delta_{K+1}^{max}$ is the maximum among $\delta_{K}(\Phi^{i})$'s and is increased with $L > 1$.
However, when $\epsilon_{1}=1$, $\delta_{K+1}^{max} < \frac{1}{\sqrt{K} +2}$ is one of solutions to satisfy (\ref{eq: sc of SOMP-a-c}) in Theorem \ref{thm: SOMP-a-c}.
In fact, Theorem \ref{thm: SOMP-a-c} requires that the mean of $\delta_{K+1}(\Phi^{i})$'s instead of $\delta_{K+1}^{max}$ is small.

Second, as shown in Theorem \ref{thm: SOMP-b-c}, $\epsilon_{2}$ indicates that $x^{i}$'s should be distributed around the center $x^{*}$, which should be far away from the origin.
In other words, $x^{i}$'s have the same sign to maximize the denominator of $\epsilon_{2}$.
To fairly compare Theorem \ref{thm: SOMP-a-c} and Theorem \ref{thm: SOMP-b-c}, we need to build the relationship between $\delta_{K}(A)$ and $\delta_{K}(\Phi^{i})$.
In fact, $\sqrt{L}\delta_{K}(A) \sim  \delta_{K+1}^{max}$.
In addition, a random matrix is known to satisfy $\delta_{cK}< \theta$ with high probability provided one chooses $M = O(\frac{cK}{\theta^{2}} \log \frac{N}{K})$ \cite{Jain2011}.
Then, it is trivial to check that when $\epsilon_{1}=1$, $\epsilon_{2}=0$ (the best case for both theorems), and $L=K$, the number of measurements required in Theorem \ref{thm: SOMP-a-c} is about $O(K)$ larger than that in Theorem \ref{thm: SOMP-b-c}.

Finally, we note that the desired signal ensembles for both the cases of $\epsilon_2$ and $\epsilon_{3}$ are the same.
Since the numerator in $\epsilon_{3}$ is fixed, it implies that $\epsilon_{2} \leq \epsilon_{3}$.
However, it should be noted that, when $\epsilon_{2}=\epsilon_{3}=0$, the sufficient condition in Theorem \ref{thm: SOMP-b-d}, compared with that in Theorem \ref{thm: SOMP-b-c}, is slightly relaxed.
In addition, the assumption in Theorem \ref{thm: SOMP-b-d} only requires that $A$, instead of all $\Phi^{i}$'s, satisfies RIP and that $K\leq M$.
Thus, even though individual $\Phi^{i}$ does not satisfy RIP, perfect reconstruction is still possible.
The following corollary shows that $K\leq M$ can be further removed for perfect support detection.
\begin{Corr}
\label{prop: SOMP-b-d-another ver}
Suppose $A$ satisfies RIP.
Then. the SOMP-(b+d) algorithm will perfectly detect the support set of $x^{i}$'s for $i=1,\ldots,L$ with the same sufficient condition in Theorem \ref{thm: SOMP-b-d}.
\end{Corr}
\begin{proof}
Detailed proof is skipped due to limited space.
\end{proof}


\section{Experiments}\label{sec:our exp}
In this section, we validate our three theorems from empirical simulations.
We first generate four types of signal ensembles as follows:
 \begin{itemize}
\item[I.] $x^{i} \sim \mathcal{N}(0,I)$ with $i=1,\ldots,L$.
\item[II.] $x^{i} \sim abs(\mathcal{N}(0,I))$ with $i=1,\ldots,L$.
\item[III.] $x^{i} \sim \mathcal{N}(I,0.25I)$ with $i=1,\ldots,L$.
\item[IV.] $x^{i} = 1$ with $i=1,\ldots,L$.
\end{itemize}
Then, we repeat the following verification procedure $100$ times for each set of parameters, composed of $M$, $K$, and $L$, under $N=100$.
\begin{enumerate}
  \item Construct $x^{i}$'s according to one of the above four types.
  \item Draw $L$ standard normal matrices $\Phi^{i} \in \mathbb{R}^{M \times N}$ for $i=1,\ldots,L$ to sample signals based on Eq. (\ref{Eq: DCS}).
  \item Run SOMP-(a+c), SOMP-(b+c), and SOMP-(b+d), respectively, to obtain an optimal point $\bar{x}^{i}$'s.
  \item Declare success if $ \sum_{i=1}^{L}\|\bar{x}^{i} - x^{i}\| \leq 10^{-5}$.
\end{enumerate}
So, the successful probability is defined as the number of successes divided by $100$.

These types of signals are designed in order to present different values of $\epsilon_{1}$, $\epsilon_{2}$ and $\epsilon_{3}$.
For example, $\epsilon_{2}$ and $\epsilon_{3}$ are gradually decreased from Type I to Type IV.
In addition, $\epsilon_{1}$'s in Types I and II are the same but are smaller than those in Types III and IV.

The results for Types I, II, III, and IV are shown in Figs. \ref{fig:Performance with L3}(a)-(d) with $L=3$, respectively, where the curve denotes the phase transition of the probability of success achieving $50\%$.
It should be noted that SOMP-(b+d)-supp only considers the success of ``support detection'' instead of signal reconstruction in SOMP-(b+d).
Thus, according to Corollary \ref{prop: SOMP-b-d-another ver}, success may happen even when $K>M$.

It is also observed from Figs. \ref{fig:Performance with L3}(a)-(c) that the curve of  SOMP-(b+d) overlaps with that of SOMP-(b+d)-supp.
This is because correct support detection implies perfect reconstruction for $K \leq M$.
In addition, it is surprising to see from Figs. \ref{fig:Performance with L3}(c)-(d) that SOMP-(b+d)-supp exhibits higher probability of success when $K$ approaches $N$.
This may be due to the fact that since the number $\binom{N}{K}$ of candidate support sets approaches $1$.
For example, when $K=N$, support detection always is correct with $\binom{N}{K}=1$.
Fig. \ref{fig:Performance with L9} reaches the same conclusions with Fig. \ref{fig:Performance with L3} but exhibits higher successful probability under $L=9$.

In summary, SOMP-(a+c) has the weakest assumption about signal ensembles such that it can be applied to all different types of signal ensembles.
Even so, for Types II-IV, its performance is not the best among the methods used for comparisons.
In fact, when $\epsilon_{1}$, $\epsilon_{2}$, and $\epsilon_{3}$ are the best choices such that the sufficient conditions are easy to satisfy in Theorems \ref{thm: SOMP-a-c}$\sim$\ref{thm: SOMP-b-d}, respectively, the sufficient condition for Theorem \ref{thm: SOMP-a-c} is relatively not easy to satisfy.

In contrast, SOMP-(b+c) outperforms SOMP-(a+c) remarkably when signals have the same sign, as shown in from Figs. \ref{fig:Performance with L3}(b)-(d).
Compared with SOMP-(b+c), the assumption in SOMP-(b+d) is more sensitive to Euclidean distances between signals, implying  large $\epsilon_{3}$, such that its performance is worse than SOMP-(b+c).
However,  in terms of support detection, SOMP-(b+d) has potential to lower the number of measurements when $M < K$.
In addition, when $\epsilon_{2}=\epsilon_{3}=0$, SOMP-(b+d) outperforms SOMP-(b+c).

\begin{figure}[t]
\begin{minipage}[b]{.49\linewidth}
  \centering{\epsfig{figure=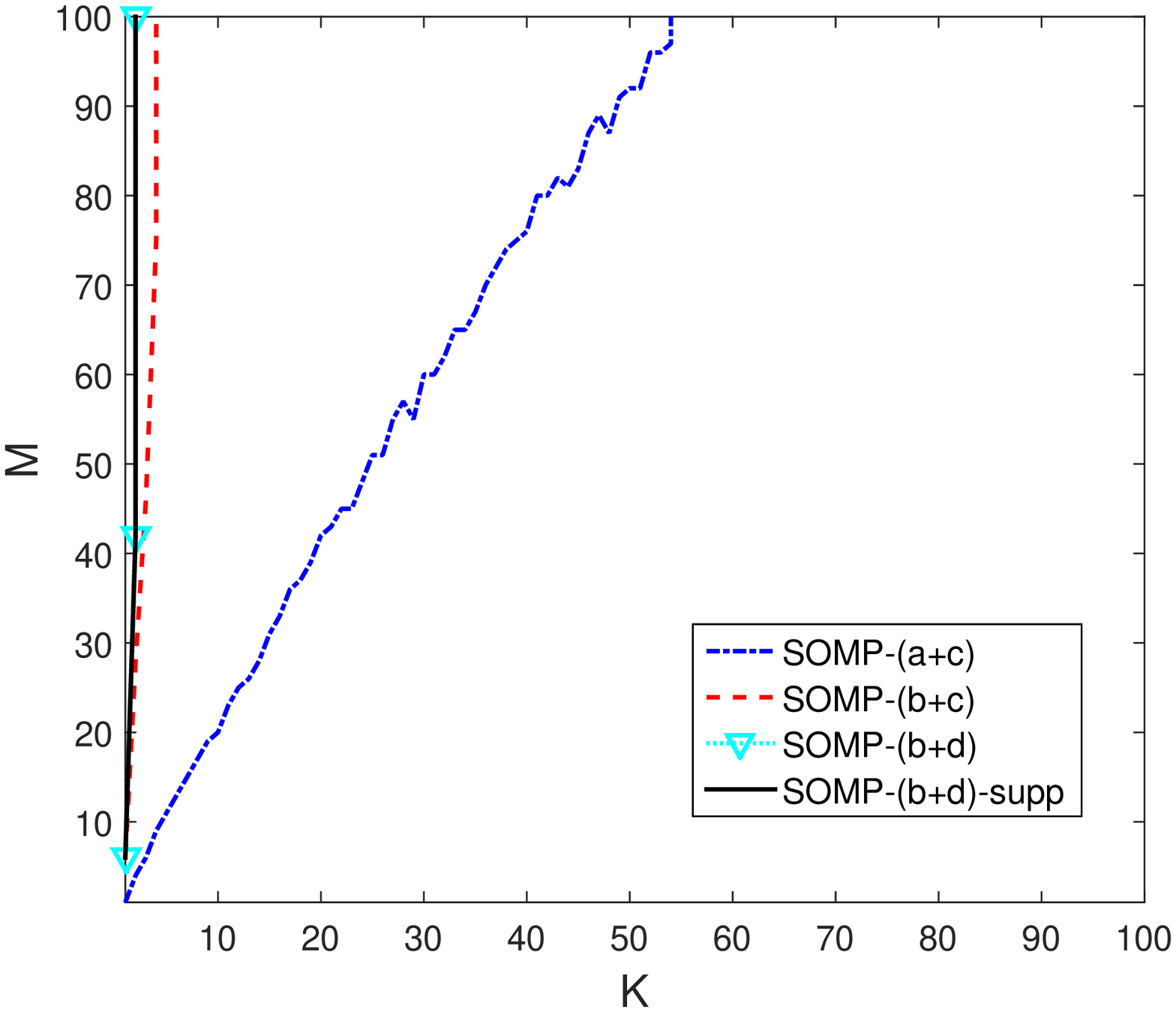,width=1.8in}}
  \centerline{(a)}
\end{minipage}
\begin{minipage}[b]{.5\linewidth}
  \centering{\epsfig{figure=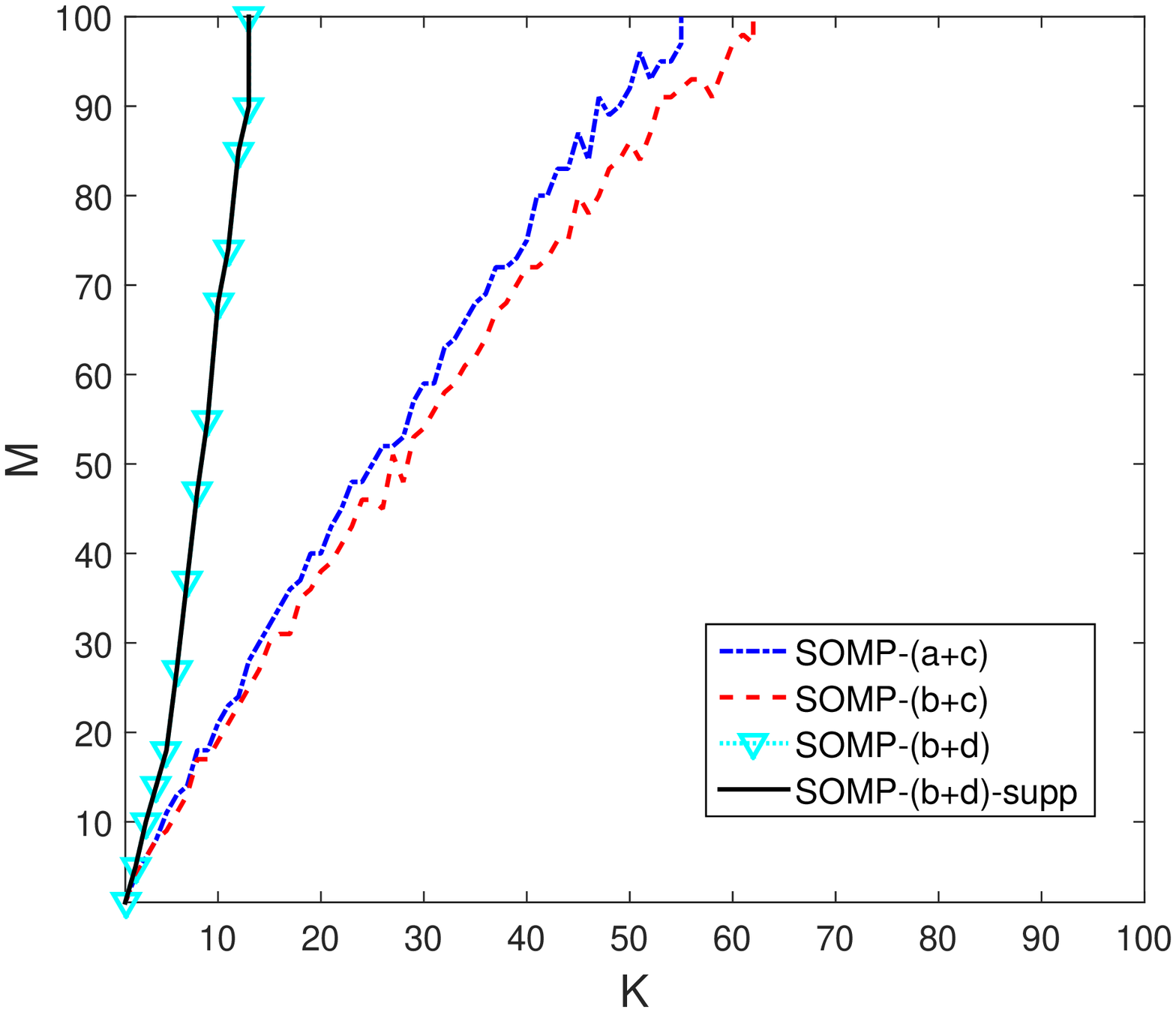,width=1.8in}}
  \centerline{(b)}
\end{minipage}
\begin{minipage}[b]{.49\linewidth}
  \centering{\epsfig{figure=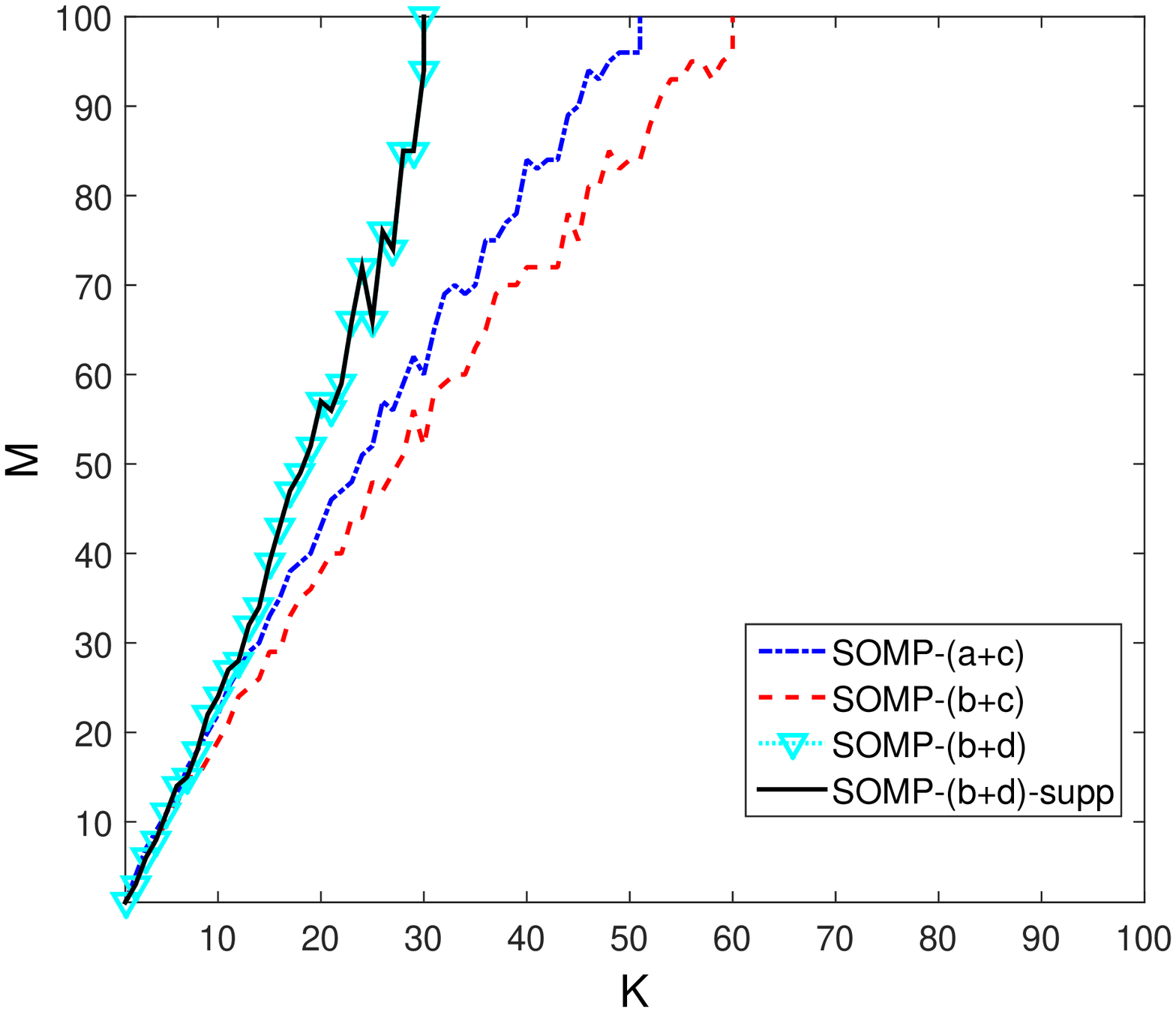,width=1.8in}}
  \centerline{(c)}
\end{minipage}
\begin{minipage}[b]{.5\linewidth}
  \centering{\epsfig{figure=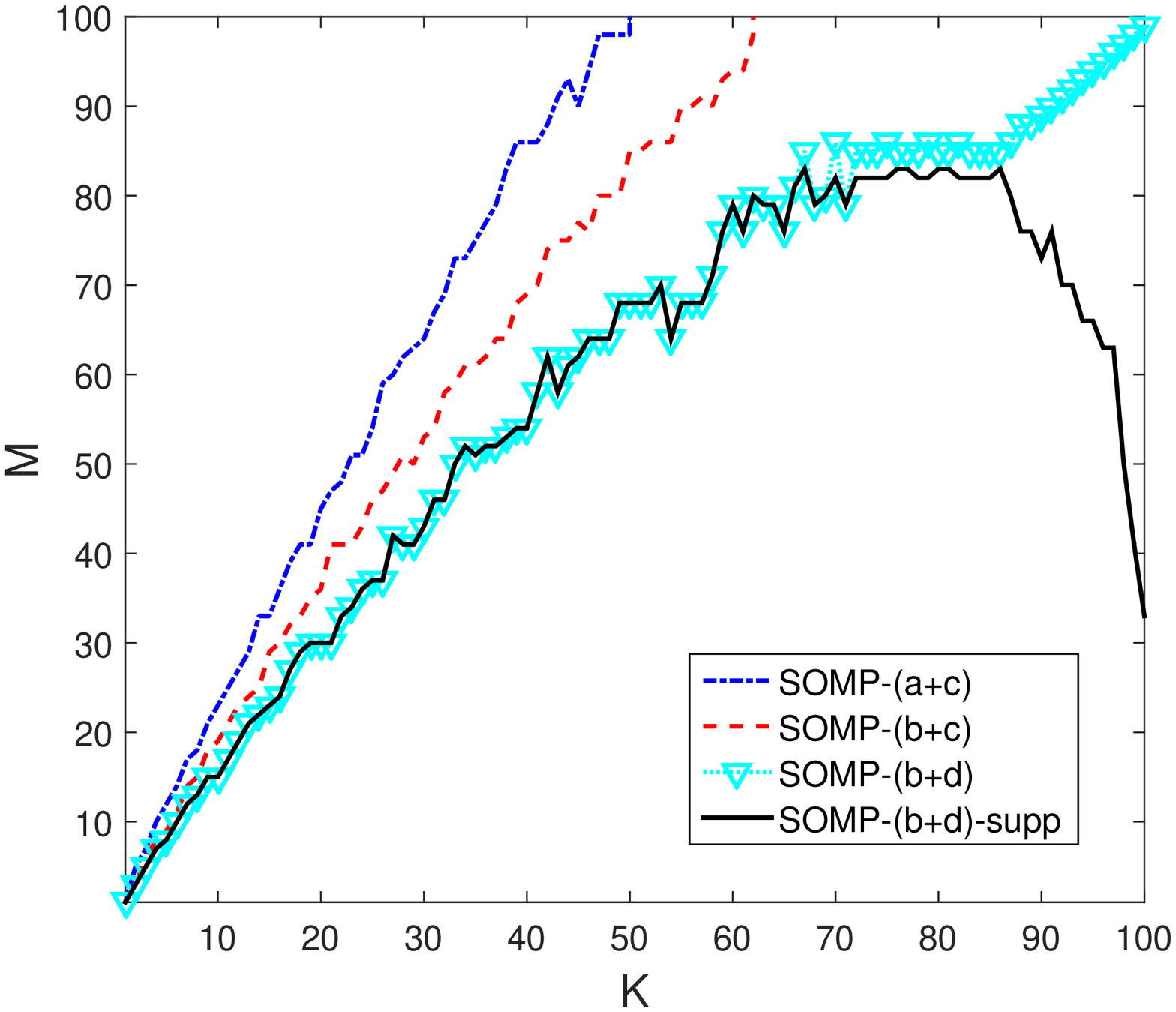,width=1.8in}}
  \centerline{(d)}
\end{minipage}
\caption{Performance analysis for different types of signals: (a) Type I; (b) Type II; (c) Type III; (d) Type IV, under $L=3$ and $N=100$.
The curve denotes the phase transition of probability of success achieving $50\%$.
The region above the curve means the probability $\geq 50\%$.}
\label{fig:Performance with L3}
\end{figure}

\begin{figure}[t]
\begin{minipage}[b]{.49\linewidth}
  \centering{\epsfig{figure=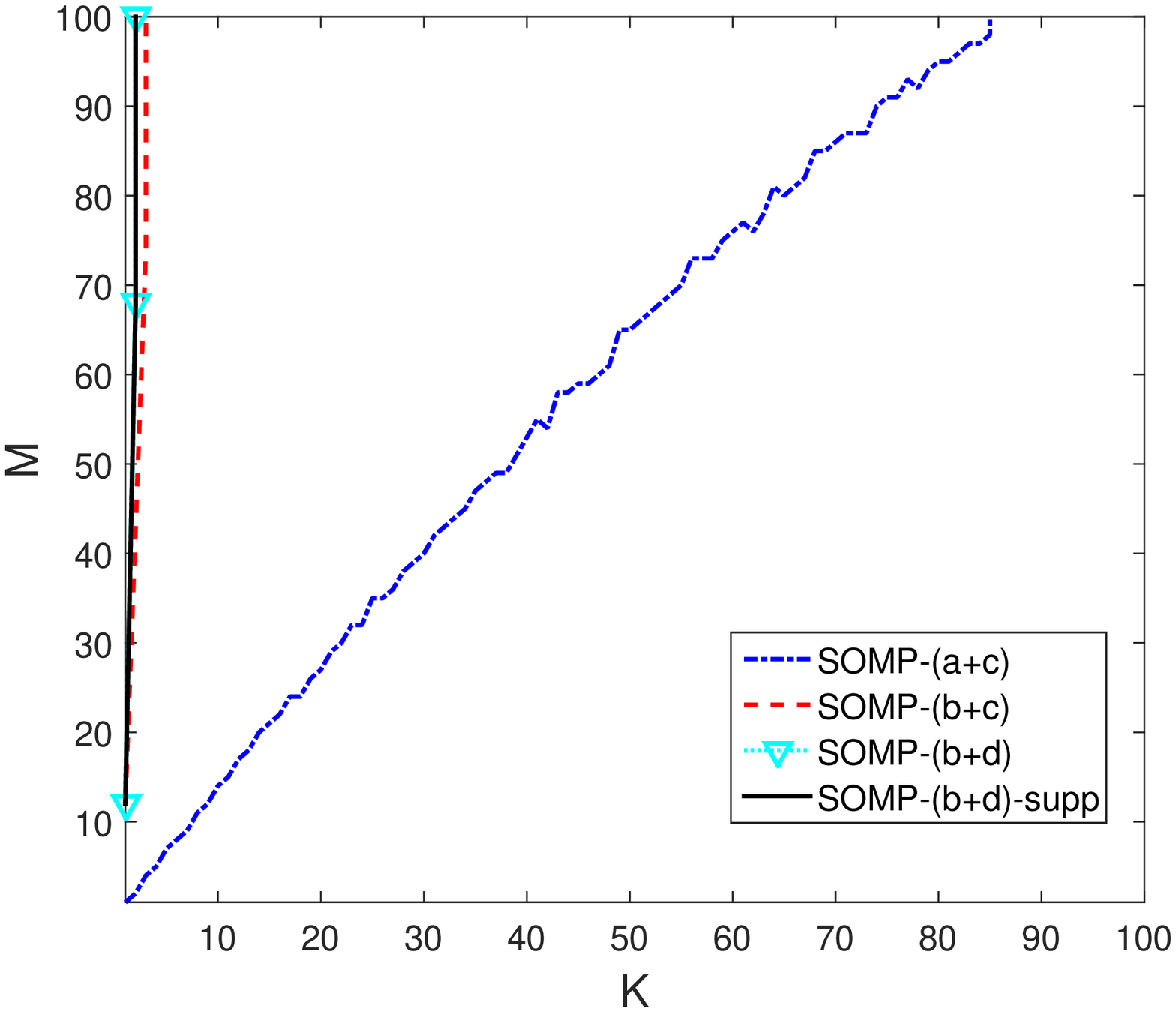,width=1.8in}}
  \centerline{(a)}
\end{minipage}
\begin{minipage}[b]{.5\linewidth}
  \centering{\epsfig{figure=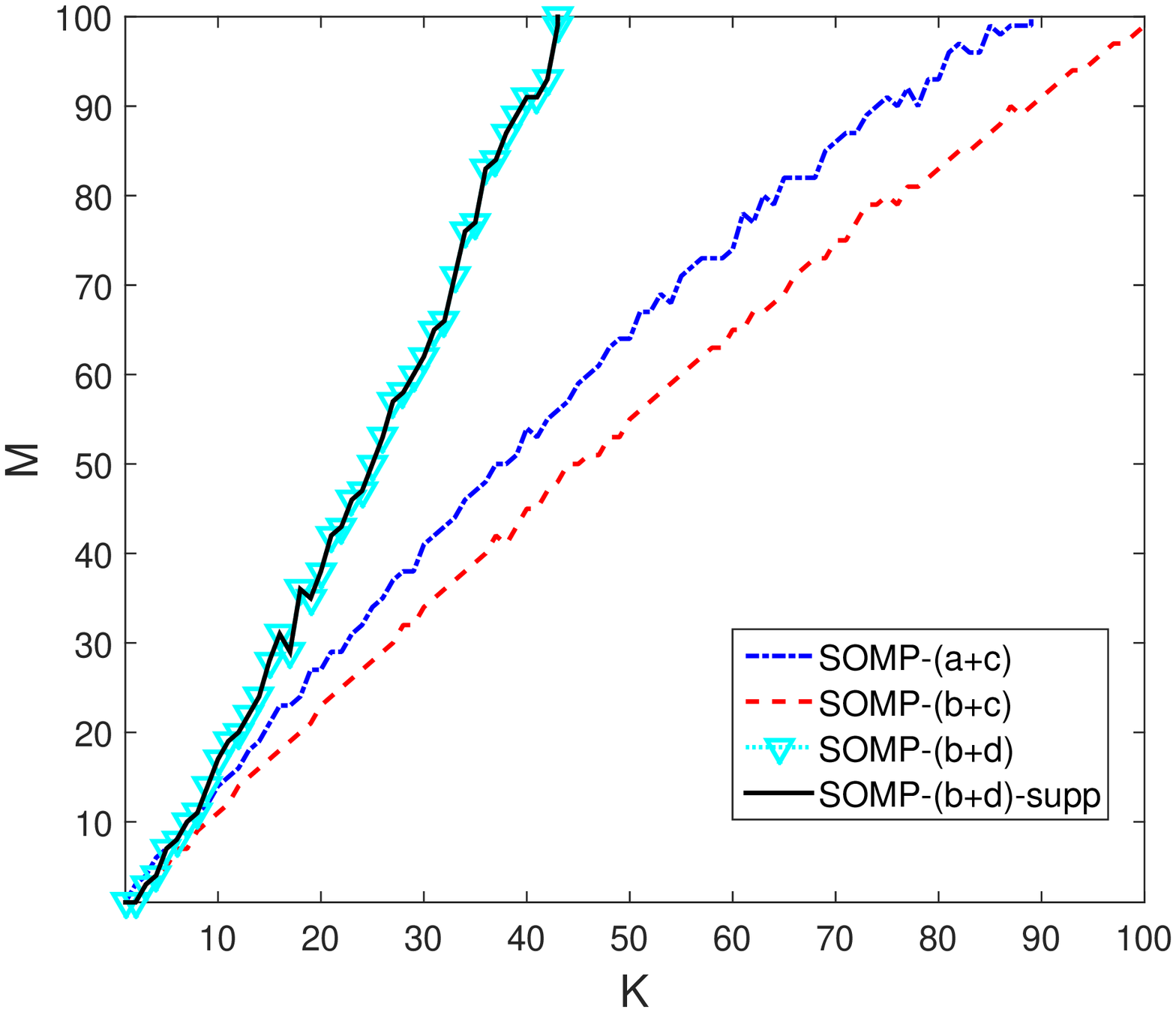,width=1.8in}}
  \centerline{(b)}
\end{minipage}
\begin{minipage}[b]{.49\linewidth}
  \centering{\epsfig{figure=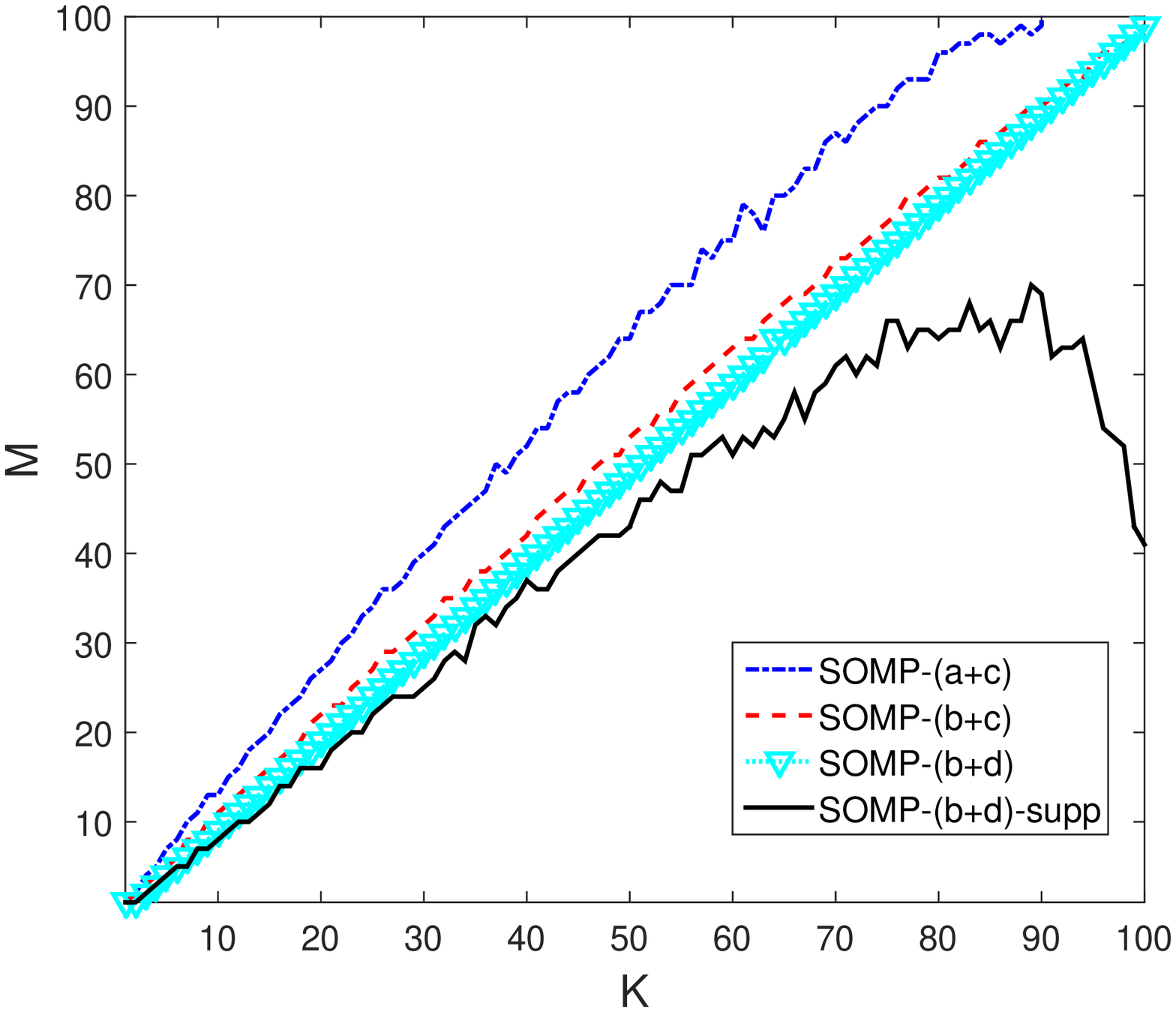,width=1.8in}}
  \centerline{(c)}
\end{minipage}
\begin{minipage}[b]{.5\linewidth}
  \centering{\epsfig{figure=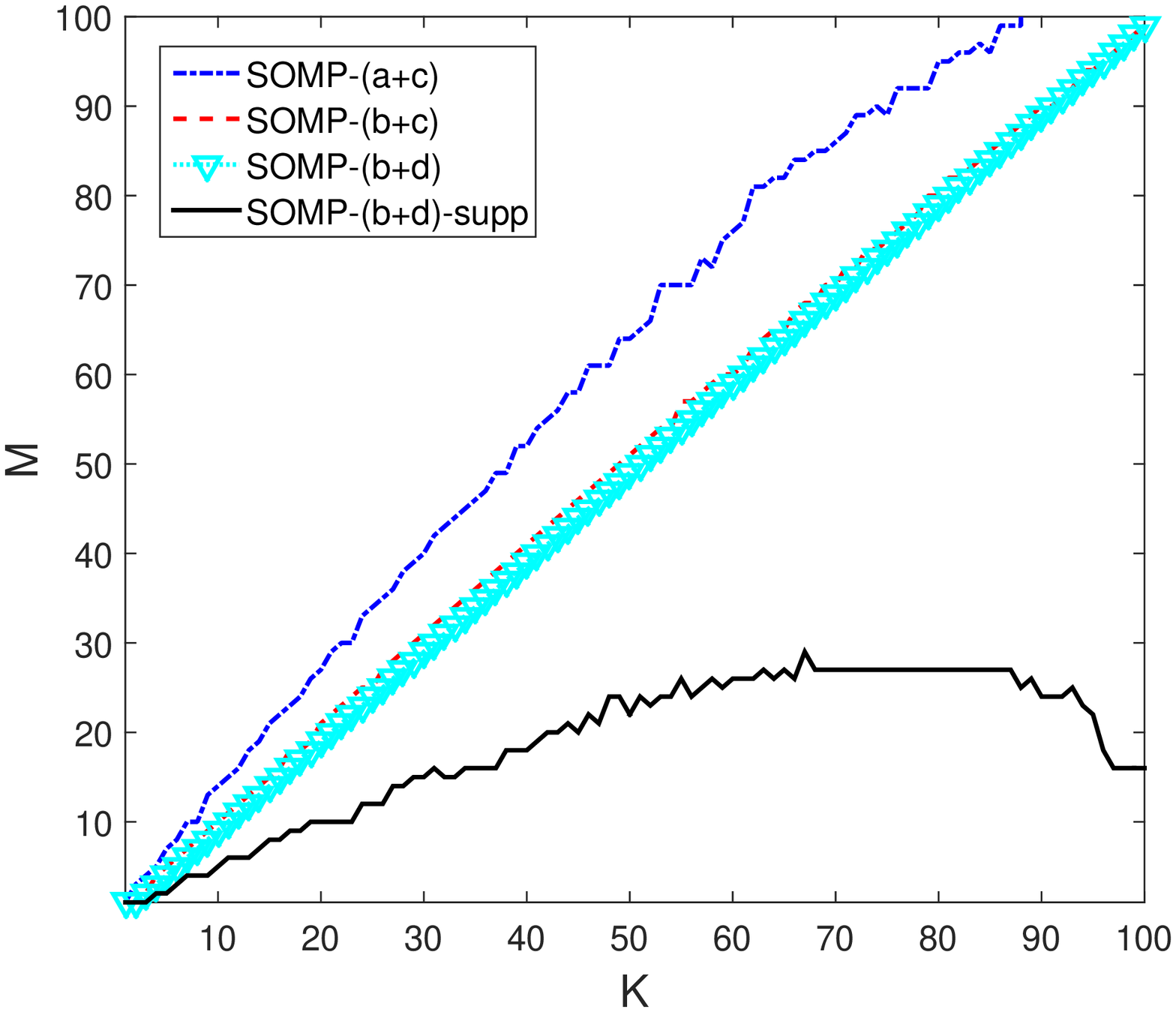,width=1.8in}}
  \centerline{(d)}
\end{minipage}
\caption{Performance analysis for different types of signals: (a) Type I; (b) Type II; (c) Type III; (d) Type IV, under $L=9$ and $N=100$.}
\label{fig:Performance with L9}
\end{figure}

\section{Appendix}\label{sec:append}
\noindent \textbf{Proof of Theorem \ref{thm: SOMP-a-c}:}
To prove Theorem \ref{thm: SOMP-a-c}, we need the following lemmas.
\begin{Lem}
Let $C = [A \mid B] \in \mathbb{R}^{m \times (k_1+k_2)}$ with $A \in \mathbb{R}^{m \times k_1}$, and $B \in \mathbb{R}^{m \times k_2}$.
Then
$$\sigma_{\min}(A^T(I-BB^{\dagger})A) \geq \sigma_{\min}(C^TC).$$
\end{Lem}
\begin{proof}
Let $B=USV^T$ with $U$ being written as $[U_I \mid U_{I^C}]$, where $I$ denotes the set of indices corresponding to non-zero singular values and $I^C$ is complement of $I$.
Since
$$
A^T(I-BB^{\dagger})A=A^T(U_{I^C}U_{I^C}^T)A=(U_{I^C}^TA)^T(U_{I^C}^TA),
$$
it remains to show that $\sigma_{\min}(U_{I^C}^TA) \geq \sigma_{\min}(C).$\\
Let $x$ be a nonzero singular vector with respect to $\sigma_{\min}(U_{I^C}^TA)$, we have
\begin{equation}
\left\| U_{I^C}^TAx \right\| = \sigma_{\min}(U_{I^C}^TA) \left\| x \right\|.
\label{lemeq1}
\end{equation}
Note that $\left\| U_{I^C}U_{I^C}^TAx \right\| = \left\| U_{I^C}^TAx \right\| $ due to the column orthogonality of $U_{I^C}$. Thus,
\begin{equation}
\left\| U_{I^C}^TAx \right\| = \left\| U_{I^C}U_{I^C}^TAx \right\| = \left\| (I-U_IU_I^T)Ax \right\|.
\label{lemeq2}
\end{equation}
Now, we can choose $z \in \mathbb{R}^{k_2}\backslash \{0\}$ such that $Bz = U_IU_I^TAx$ by using the fact that span($B$) = span($U_I$).
Let $v = \begin{bmatrix} x\\ -z  \end{bmatrix}$, we have
$$
\begin{aligned}
\left\| Cv \right\| &= \left\|  Ax-Bz \right\| = \left\|  Ax-U_IU_I^TAx \right\| \\
&\stackrel{(\ref{lemeq2})}{=}  \left\|  U_{I^C}^TAx \right\| \stackrel{(\ref{lemeq1})}{=} \sigma_{\min}(U_{I^C}^TA) \left\| x \right\|.
\end{aligned}
$$
On the other hand, we have $\left\| Cv \right\| \geq \sigma_{\min}(C)\left\| v \right\|$.
Therefore, $$\sigma_{\min}(C) \leq \sigma_{\min}(U_{I^C}^TA)\frac{\|x\|}{\|v\|} \leq \sigma_{\min}(U_{I^C}^TA),$$
and we completed the proof.
\end{proof}

\begin{Lem}(\textbf{Monotonicity of the RIC})\cite{Candes2005}\\
If the sensing matrix $\Phi$ satisfies the RIP of both orders $K_1$ and $K_2$, then $\delta_{K_1}(\Phi) \leq \delta_{K_2}(\Phi)$ for any $K_1 \leq K_2$.
\end{Lem}

\begin{Lem}\cite{Candes2008}
\label{lemma: disjoint}
Let $I_1$, $I_2 \subset \Omega$ be two disjoint sets ($I_1 \cap I_2 =\emptyset$). If $\delta_{|I_1|+|I_2|} < 1,$ then
$$\|(\Phi_{I_1})^{T}\Phi_{I_2}x\|_2 \leq \delta_{|I_1|+|I_2|}(\Phi)\|x\|_2$$
holds for any $x$.
\end{Lem}


Now, we are ready to prove Theorem \ref{thm: SOMP-a-c}.
Here, $I_{1}$ and $I_{2}$ in Lemma \ref{lemma: disjoint} denote the chosen index at $t$-th iteration and ground truth $\Omega$, respectively.
Thus, by the contrapositive of statement in Lemma \ref{lemma: disjoint}, it implies that the chosen index and $\Omega$ are not disjoint; {\em i.e.}, support detection is correct.


\begin{proof}
Let $\Phi^j_{\Omega} = [\Phi^j_S \mid \Phi^j_U]$, where $S$ denotes the support set that has been solved and $U$ denotes the support set that has not solved yet. Let $r^{j,t}$ be the residual and let $I$ be the chosen index at $t$-th iteration. For simplicity, let $r^j$ and $\delta^j$ denote $r^{j,t}$ and $\delta(\Phi^j)$, respectively.
When $I \notin \Omega$, we first derive the upper bound of $ \sum_j \left\| (\Phi^j_I)^Tr^{j} \right\|$ by Lemma \ref{lemma: disjoint} as:
$$
\begin{aligned}
\label{eq: the upper bound of SOMP-a+c}
\displaystyle \sum_j & \left\| (\Phi^j_I)^Tr^j\right\| = \sum_j\left\| (\Phi^j_I)^T(I-\Phi^j_S(\Phi^j_S)^{\dagger})\Phi^j_Ux^j_U\right\| \\
 &\leq  \displaystyle \sum_j \left( \left\| (\Phi^j_I)^T\Phi^j_Ux^j_U \right\| + \left\| (\Phi^j_I)^T\Phi^j_S(\Phi^j_S)^{\dagger}\Phi^j_Ux^j_U \right\| \right)\\
&= \displaystyle \sum_j \left[ (I) + (II) \right],
\end{aligned}
$$
where
$$
\begin{array}{ccl}
(I) &\leq& \delta^j_{1+|U|}\left\| x^j_U\right\| \leq \delta^j_{K+1}\left\| x^j_U\right\|\\
(II) & \leq& \displaystyle\frac{(\delta^j_{K+1})^2}{1-\delta^j_{K+1}}\left\| x^j_U\right\| (\cite{Kwon2014}).
\end{array}
$$
\end{proof}
Plugging $(I)$ and $(II)$ into above inequality, we can obtain:
\begin{equation}
\label{eq: the upper bound of SOMP-a+c}
\displaystyle \sum_j\left\| (\Phi^j_I)^Tr^j\right\| \leq \displaystyle\sum_j\left(\delta^j_{K+1}+\frac{(\delta^j_{K+1})^2}{1-\delta^j_{K+1}} \right)\max_j\left\| x^j_U\right\|.
\end{equation}

Then, we derive the lower bound of $\displaystyle \sum_j \left\| (\Phi^j_I)^Tr^j\right\|$ as follows:
\begin{equation}
\begin{aligned}
\left\| (\Phi^j_I)^Tr^j\right\| &= \max_{i \in I}\left| \langle \phi^j_i,r^j\rangle\right| \geq \displaystyle\sqrt{\frac{1}{|\Omega|}\sum_{i \in \Omega}\left|\langle\phi^j_i,r^j\rangle\right|^2}\\
&=\frac{1}{\sqrt{K}}\left\|(\Phi^j_{\Omega})^Tr^j\right\|.
\end{aligned}
\end{equation}
Then, we have :
\begin{equation}
\label{eq: the lower bound of SOMP-a+c}
\begin{aligned}
\displaystyle & \sum_j \sqrt{K}\left\| (\Phi^j_I)^Tr^j\right\| \\
& \geq\displaystyle\sum_j\left\|(\Phi^j_{\Omega})^Tr^j\right\| =\displaystyle\sum_j\left\| (\Phi^j_U)^T(I-\Phi^j_S(\Phi^j_S)^{\dagger}\Phi^j_U)x^j_U\right\| \\
&\geq\displaystyle \sum_j\left(1-\delta^j_{K+1}\right)\left\|x^j_U\right\| \geq \sum_j\left(1-\delta^j_{K+1}\right)\min_j\left\|x^j_U\right\|.
\end{aligned}
\end{equation}
Finally, we want (\ref{eq: the lower bound of SOMP-a+c})$>$(\ref{eq: the upper bound of SOMP-a+c}) as: \\
$$
\begin{aligned}
\displaystyle\frac{1}{\sqrt{K}}\sum_j\left(1-\delta^j_{K+1}\right)& \min_j\left\|x^j_U\right\|  >  \\  &\displaystyle\sum_j\left(\delta^j_{K+1}+\frac{(\delta^j_{K+1})^2}{1-\delta^j_{K+1}}\right)\max_j\left\|x^j_U\right\|,
\end{aligned}
$$
which implies that
$$
\displaystyle \sum_j \frac{\epsilon_1\left(\delta^j_{K+1}\right)^2-\left(2\epsilon_{1}+\sqrt{K}\right)\delta^j_{K+1}+\epsilon_{1}}{1-\delta^j_{K+1}}>0,
$$
where $\epsilon_{1}=\displaystyle\max_{U\in \mathcal{P}(\Omega)\backslash\emptyset}\frac{min_j\|x^j_U\|}{\max_j\|x^j_U\|}$.
Note that the proof is independent of the iteration index $t$, and hence the condition holds at each iteration.
Thus, given correct support, the SOMP-(a+c) algorithm will perfectly reconstruct $x^i$'s.\\

\noindent \textbf{Proof of Theorem \ref{thm: SOMP-b-c}:}
\begin{proof}
The proof is similar to Theorem \ref{thm: SOMP-a-c} and we use the same notations in this proof.
We need to derive the lower bound and upper bound of $\left\| \sum_j (\Phi^j_I)^Tr^j \right\|$. First, the lower bound is shown as follows:\\
$$
\begin{aligned}
&\displaystyle \sqrt{K}\left\| \sum_j (\Phi^j_I)^Tr^j \right\| =\displaystyle \left\| \sum_j(\Phi^j_U)^T(I-\Phi^j_S(\Phi^j_S)^{\dagger})\Phi^j_Ux^j_U\right\| \\
&\displaystyle \geq \left\| \sum_j(\Phi^j_U)^T\Phi^j_Ux^j_U \right\| -  \left\| \sum_j(\Phi^j_U)^T\Phi^j_S(\Phi^j_S)^{\dagger}\Phi^j_Ux^j_U \right\| \\
&= (I)+(II).
\end{aligned}
$$
Let $x^j_U = x^*_U - c^j$, where $c^j \in \mathbb{R}^{|U|}$ is any constant and $\displaystyle x^*_U = \frac{\sum_j x^j_U}{L}$, and let $\displaystyle \epsilon_{2}=\max_{U\in \mathcal{P}{(\Omega)}\backslash \emptyset}\frac{\sum_j\|x^j_U - x^*_U\|}{L\|x^*_U\|}$. Then, we have
$$
\begin{array}{cll}
(I) &=& \displaystyle \left\| \sum_j\left[ (\Phi^j_U)^T\Phi^j_U-I+I\right]x^j_U\right\|  \vspace{+3pt} \\
&   \geq& \displaystyle \left\| \sum_jx^j_U \right\| - \left\| \sum_j\left[ (\Phi^j_U)^T\Phi^j_U-I \right](x^*_U-c^j)\right\| \\
&\geq& \displaystyle L\left\|x^*\right\| - \left\| \sum_j\left[ (\Phi^j_U)^T\Phi^j_U-I \right]x^*_U\right\| - \\
&\ & \displaystyle \left\| \sum_j\left[ (\Phi^j_U)^T\Phi^j_U-I \right]c^j \right\|.
\end{array}
$$
Since
$$\displaystyle\left\| \sum_j\left[ (\Phi^j_U)^T\Phi^j_U-I \right]x^*_U \right\|\leq L\delta_{K+1}(A)\left\|x^*_U\right\|$$
and
$$
\begin{aligned}
&\displaystyle \left\| \sum_j\left[ (\Phi^j_U)^T\Phi^j_U-I \right]c^j \right\| \\
&= \displaystyle \left\| \sum_j\left[ (\Phi^j_U)^T\Phi^j_U-I \right](x^*_U-x^j_U) \right\| \\
&\leq \displaystyle L\epsilon_{2}\left\|x^*_U\right\|\sum_j\delta^j_{|U|} \leq \displaystyle L\epsilon_{2}\left\|x^*_U\right\|\sum_j\delta^j_{K+1},
\end{aligned}
$$
we have:
$$
\begin{array}{cll}
(I) &\geq&\displaystyle L\left\|x^*_U\right\|\left(1-\delta_{K+1}(A)-\epsilon_{2}\sum_j\delta^j_{K+1}\right)\\
&\geq& \displaystyle L\left\|x^*_U\right\|\left(1-\delta_{K+1}(A)-L\epsilon_{2}\delta^{max}_{K+1}\right).
\end{array}
$$
In addition,
$$
\begin{array}{cll}
(II) &\leq& \displaystyle \sum_j \frac{\left(\delta^j_{|U|+|S|}\right)}{1-\delta^j_{|S|}}(1+L\epsilon_{2})\left\|x^*_U\right\| \\
&\leq& \displaystyle \sum_j \frac{\left(\delta^j_{K+1}\right)^2}{1-\delta^j_{K+1}}(1+L\epsilon_{2})\left\|x^*_U\right\|. \\
&\leq& \displaystyle  L\frac{\left(\delta^{max}_{K+1}\right)^2}{1-\delta^{max}_{K+1}}(1+L\epsilon_{2})\left\|x^*_U\right\|.
\end{array}
$$\\
Therefore, the lower bound of $\left\| \sum_j (\Phi^j_I)^Tr^j \right\| $ is
\begin{equation}
\frac{L\|x^*\|}{\sqrt{K}}\left[ 1-\delta_{K+1}(A)-L\epsilon_{2}\delta^{max}_{K+1}-\frac{(\delta^{max}_{K+1})^2}{1-\delta^{max}_{K+1}} \right].
\label{lower}
\end{equation}
On the other hand, the upper bound is obtained by:
$$
\begin{aligned}
&\left\| \sum_j (\Phi^j_I)^Tr^j \right\| \\
&\leq \displaystyle\left( \left\|  \sum_j (\Phi^j_I)^T\Phi^j_Ux^j_U \right\| + \left\| \sum_j  (\Phi^j_I)^T\Phi^j_S(\Phi^j_S)^{\dagger}\Phi^j_Ux^j_U \right\| \right)\\
& = (III) + (IV),
\end{aligned}
$$
where
$$
\begin{aligned}
(III) & \leq L\left\| x^*\right\| \left(\delta_{1+|U|}(A)+\epsilon_{2}\sum_j\delta^j_{1+|U|}\right)\\
&\leq L\left\| x^*\right\| \left(\delta_{K+1}(A)+\epsilon_{2}L\delta^{max}_{K+1}\right)
\end{aligned}
$$
and
$$
\begin{aligned}
(IV) &\leq \displaystyle \sum_j \frac{\left(\delta^j_{1+|S|}\right)^2}{1-\delta^j_{|S|}}\left\| x^j_U\right\|  \\
&\leq \displaystyle \frac{L\left(\delta^{max}_{K+1}\right)^2}{1-\delta^{max}_{K+1}}(1+L\epsilon_{2})\left\| x^*_U\right\|.
\end{aligned}
$$
Therefore, the upper bound of $\left\| \sum_j (\Phi^j_I)^Tr^j \right\| $ is:
\begin{equation}
L\left\| x^*_U \right\|\left[ \delta_{K+1}(A)+L\epsilon_{2}\delta^{max}_{K+1}+\frac{\left(\delta^{max}_{K+1}\right)^2}{1-\delta^{max}_{K+1}}(1+L\epsilon_{2})\right]
\label{upper}
\end{equation}
Hence, the SOMP-(b+c) algorithm will choose correct support if (\ref{lower}) $>$ (\ref{upper}), which implies\\
$$(\sqrt{K}+1)\delta_{K+1}(A)+(1+(\sqrt{K}+1)L\epsilon_{2})\delta^{max}_{K+1}<1.$$
When all support are found correctly, the SOMP-(b+c) algorithm will perfectly reconstruct $x^i$'s.
\end{proof}

\noindent \textbf{Proof of Theorem \ref{thm: SOMP-b-d}:}
\begin{proof}
In the proof, we need to derive the lower bound and upper bound of $\left\| \sum_j (\Phi^j_I)^Tr^j \right\|$.
First, the lower bound is derived as follows:
$$
\begin{aligned}
\sqrt{K}& \left\|  \sum_j (\Phi^j_I)^Tr^j\right\|  \geq \left\| \sum_j (\Phi^j_T)^Tr^j\right\| \\
&\geq \left\| \sum_j x^j_U\right\| - \left\| \sum_j \left[(\Phi^j_T)^T\Phi^j_T-I\right]x^j_U\right\|  \\
&= L\left\| x^*_U\right\| - \left\| \sum_j \left[(\Phi^j_T)^T\Phi^j_T-I\right](x^*_U-c^j)\right\| \\
&\geq L\left\| x^*_U\right\| - \left\| \sum_j \left[(\Phi^j_T)^T\Phi^j_T-I\right]x^*_U\right\| \\
&- \left\| \sum_j \left[(\Phi^j_T)^T\Phi^j_T-I\right]c^j\right\| \\
& = L\left\| x^*_U\right\| -(I) - (II),
\end{aligned}
$$
where
$$\displaystyle (I) \leq  \sum_j\delta^j_{|T|} \left\| x^*_U\right\|\leq \sum_j\delta^j_{K+1} \left\| x^*_U\right\|$$
and \\
$$\displaystyle (II) \leq \epsilon_{3}L\left\| x^*_U\right\| \sum_j\delta^j_{|T|} \leq \epsilon_{3}L\left\| x^*_U\right\| \sum_j\delta^j_{K+1}.$$
Therefore, we can obtain
\begin{equation}
\label{thm3_lower}
\begin{aligned}
\displaystyle  \sqrt{K} \left\|  \sum_j (\Phi^j_I)^Tr^j\right\| &>  L\left\| x^*_U\right\| - \left\| x^*_U \right\| \left( \sum_j\delta^j_{K+1}(1+\epsilon_{3}L) \right) \\
 & > L\left\| x^*_U\right\| \left( 1 -  \delta^{max}_{K+1}(1+\epsilon_{3}L)  \right).
\end{aligned}
\end{equation}
On the other hand, the upper bound is derived as:
$$
\begin{aligned}
&\left\| \sum_j(\Phi^j_I)^Tr^j \right\| \\
&\leq \left\| \sum_j(\Phi^j_I)^T\Phi^j_Tx^*_U \right\| + \left\| \sum_j(\Phi^j_I)^T\Phi^j_Tc^j \right\| \\
&= (III)+(IV),
\end{aligned}
$$
where
$$(III)=L\left\| A^T_IA_Tx^*_U\right\| \leq L\delta_{K+1}(A)\left\| x^*_U \right\|$$
and
$$(IV) \leq L^3\epsilon_{3}\delta_{K+1}(A)\left\| x^*_U \right\|.$$
Finally, we can obtain
\begin{equation}
L\left\| x^*_U\right\| \delta_{K+1}(A)(1+L^2\epsilon_{3}).
\label{thm3_upper}
\end{equation}
In sum, the SOMP-(b+d) algorithm will select correct support if $(\ref{thm3_lower})>(\ref{thm3_upper})$, implying
$$\sqrt{K}(1+L^2\epsilon_{3})\delta_{K+1}(A) + (1+L\epsilon_{3}) \delta^{max}_{K+1} < 1.$$
When all support are found correctly, SOMP-(b+d) algorithm will perfectly recover $x^i$'s.
\end{proof}

\bibliographystyle{IEEEbib}	
\bibliography{refs}		

\end{document}